\documentclass{acm_proc_article-sp}
\usepackage{amssymb}
\usepackage{xspace}
\usepackage[ruled,vlined,linesnumbered,resetcount]{algorithm2e}
\usepackage[T1]{fontenc}
\usepackage{pgfplotstable}
\usepackage{color}
\date{\today}
\usepackage{etex}
\usepackage[english]{babel}
\usepackage{amsmath}
\usepackage{amssymb}
\usepackage{bibentry}
\usepackage{enumitem}
\usepackage{xytree}
\usepackage{cleveref}
\usepackage[ruled,vlined]{algorithm2e}
\usepackage{subcaption}

\newtheorem{lemma}{Lemma}
\newtheorem{corollary}{Corollary}

\renewcommand{\O}{\mathcal O}
\renewcommand{\P}{\mathbb P}

\newcommand{\bfsc}{BFSCut}
\newcommand{\preproc}{PreProcessing}
\newcommand{\bfsctt}{\texttt{\upshape \bfsc}}

\newcommand{\deltabfs}{{\sc Olh}}
\newcommand{\sampl}{{\sc Ocl}}
\newcommand{\newalg}{{\sc Bcm}}
\newcommand{\mvis}{m_{\text{vis}}}
\newcommand{\mtot}{m_{\text{tot}}}
\newcommand{\degr}[1]{\mathrm{deg}(#1)}
\newcommand{\odegr}[1]{\mathrm{outdeg}(#1)}
\newcommand{\Reac}[1]{R(#1)}
\newcommand{\reac}[1]{r(#1)}
\newcommand{\neig}[2]{\Gamma_{#1}(#2)}
\newcommand{\sneig}[2]{\gamma_{#1}(#2)}
\newcommand{\uneig}[2]{\tilde{\gamma}_{#1}(#2)}
\newcommand{\ball}[2]{N_{#1}(#2)}
\newcommand{\sball}[2]{n_{#1}(#2)}
\newcommand{\farn}[1]{f(#1)}
\newcommand{\dfarn}[2]{f_{#1}(#2)}
\newcommand{\lfarn}[3]{\tilde{f}_{#1}(#2,#3)}
\newcommand{\clos}[1]{c(#1)}
\newcommand{\uclos}[2]{\tilde{c}_{#1}(#2)}

\begin{document}

\title{Fast and Simple Computation \\ of Top-k Closeness Centralities}
\numberofauthors{3}
\author{\inst{1}, Pierluigi Crescenzi\inst{2}, Andrea Marino\inst{3}}
\numberofauthors{2}

\author{
Michele Borassi\\
       \affaddr{IMT Institute for Advanced Studies Lucca, Italy}\\
        Pierluigi Crescenzi\\
       \affaddr{Dipartimento di Ingegneria dell'Informazione, Universit\`a di Firenze}\\
 Andrea Marino\\
       \affaddr{Dipartimento di Informatica, Universit\`a di Pisa}\\
}

\maketitle

\begin{abstract}
Closeness is an important centrality measure widely used in the analysis of real-world complex networks. In particular, the problem of selecting the $k$ most central nodes with respect to this measure has been deeply analyzed in the last decade. However, even for not very large networks, this problem is computationally intractable in practice: indeed, Abboud et al have recently shown that its complexity is strictly related to the complexity of the All-Pairs Shortest Path (in short, APSP) problem, for which no subcubic ``combinatorial'' algorithm is known. In this paper, we propose a new algorithm for selecting the $k$ most closeness central nodes in a graph. In practice, this algorithm significantly improves over the APSP approach, even though its worst-case time complexity is the same. For example, the algorithm is able to compute the top $k$ nodes in few dozens of seconds even when applied to real-world networks with millions of nodes and edges. We will also experimentally prove that our algorithm drastically outperforms the most recently designed algorithm, proposed by Olsen et al. Finally, we apply the new algorithm to the computation of the most central actors in the IMDB collaboration network, where two actors are linked if they played together in a movie. 
\end{abstract}



%\keywords{Network Science, Centrality, Closeness, Complex Networks}

\pagestyle{plain}
\sloppy

\begin{table*}
\begin{center}
\caption{notations used throughout the paper.}
\label{tbl:notation}
\begin{tabular}{||p{4cm}|p{2cm}|p{10cm}||}
\hline
\textbf{Name} & \textbf{Symbol} & \textbf{Definition}\\
\hline\hline
Graphs & $G=(V,E)$ & Graph with node/vertex set $V$ and edge/arc set $E$\\
\cline{2-3}
& $n$ & $|V|$\\
\cline{2-3}
& $\mathcal{G=(V,E)}$ & Weighted DAG of strongly connected components (see Section~\ref{sec:defs})\\
\hline
Degree functions & $\degr v$ & Degree of a node in an undirected graph\\
\cline{2-3}
& $\odegr v$ & Out-degree of a node in a directed graph\\
\hline
Distance function & $d(v,w)$ & Number of edges in a shortest path from $v$ to $w$\\
\hline
Reachability set function & $\Reac v$ & Set of nodes reachable from $v$ (by definition, $v\in \Reac v$)\\
\cline{2-3}
& $\reac v$ & $|\Reac v|$\\
\cline{2-3}
& $\alpha(v)$ & Lower bound on $\reac v$, that is, $\alpha(v) \leq \reac v$ (see Section~\ref{sec:alphaomega})\\
\cline{2-3}
& $\omega(v)$ & Upper bound on $\reac v$, that is, $\reac v \leq \omega(v)$ (see Section~\ref{sec:alphaomega}) \\
\hline
Neighborhood functions & $\neig dv$ & Set of nodes at distance $d$ from $v$, that is, $\{w \in V: d(v,w)=d\}$\\
\cline{2-3}
& $\sneig dv$ & Number of nodes at distance $d$ from $v$, that is, $|\neig dv|$\\
\cline{2-3}
& $\uneig {d+1}v$ & Upper bound on $\sneig {d+1}v$, defined as $\sum_{u \in \neig dv} \degr u$ if the graph is undirected, $\sum_{u \in \neig dv}\odegr u$ otherwise (clearly, $\sneig {d+1}v \leq \uneig {d+1}v$)\\
\hline
Ball functions & $\ball dv$ & Set of nodes at distance \textit{at most} $d$ from $v$, that is, $\{w \in V: d(v,w)\leq d\}$\\
\cline{2-3}
& $\sball dv$ & Number of nodes at distance \textit{at most} $d$ from $v$, that is, $|\ball dv|$\\
\hline
Farness functions & $\farn v$ & Farness of node $v$, that is, $\sum_{w \in \Reac v} d(v,w)$\\
\cline{2-3}
& $\dfarn dv$ & Farness of node $v$ up to distance $d$, that is, $\sum_{w \in \ball dv} d(v,w)$\\
\cline{2-3}
& $\lfarn dvx$ & Lower bound function on farness of node $v$, that is, $\dfarn dv-\uneig{d+1}v+(d+2)(x-\sball dv)$ (see Lemma~\ref{lem:boundokconn})\\
\hline
Closeness functions & $\clos v$ & Closeness of node $v$, tha is, $\frac{(\reac v - 1)^2}{(n-1)\farn v}$\\
\cline{2-3}
& $\uclos dv$ & Upper bound on closeness of node $v$, that is, $\frac{(\reac v - 1)^2}{(n-1)\lfarn dv{\reac v}}$ (see Corollary~\ref{cor:ubound})\\
\hline
\end{tabular}

\end{center}
\end{table*}

\section{Introduction}
The problem of identifying the most central nodes in a network is a fundamental question that has been asked many times in a plethora of research areas, such as biology, computer science, sociology, and psychology. Because of the importance of this question, several centrality measures have been introduced in the literature (for a recent survey, see~\cite{Boldi2014}). Closeness is certainly one of the oldest and of the most widely used among these measures~\cite{Bavelas1950}. Informally, the closeness of a node $v$ is the inverse of the expected distance between $v$ and a random node $w$, and it somehow estimates the efficiency of a node in spreading information to other nodes: the larger the closeness, the most ``influential'' the node is.

Formally, given a (directed strongly) connected graph $G=(V,E)$, the \textit{closeness} of a vertex $v$ is defined as $\clos v=\frac{n-1}{\farn v}$, where $n=|V|$, $\farn v = \sum_{w \in V} d(v,w)$ is the \textit{farness} of $v$, and $d(v,w)$ is the distance between the two vertices $v$ and $w$ (that is, the number of edges in a shortest path from $v$ to $w$). If $G$ is not (strongly) connected, the definition is more complicated, because $d(v,w)$ is defined only for vertices reachable from $v$. Let $\Reac v$ be the set of these vertices, and let $\reac v$ denote its cardinality (note that $v \in \Reac v$ by definition). In the literature \cite{Lin1976,Wasserman1994,Olsen2014}, the most common generalization is
\[\clos v =\frac{\reac v - 1}{\farn v}\frac{\reac v-1}{n-1}=\frac{(\reac v - 1)^2}{(n-1)\farn v}\]
where $\farn v=\sum_{w \in \Reac v} d(v,w)$.

In order to compute the $k$ vertices with largest closeness, the textbook algorithm computes $\clos v$ for each $v$ and returns the $k$ largest found values. The main bottleneck of this approach is the computation of $d(v,w)$, for each pair of vertices $v$ and $w$ (that is, solving the All-Pairs Shortest Paths or APSP problem). This can be done in two ways: either by using fast matrix multiplication, in time $\O(n^{2.373})$ \cite{Williams2012}, or by performing a breadth-first search (in short, BFS) from each vertex $v \in V$, in time $\O(mn)$, where $m=|E|$. Usually, the BFS approach is preferred because the other approach contains big constants hidden in the $\O$ notation, and because real-world networks are usually sparse, that is, $m$ is not much bigger than $n$. However, also this approach is too time-consuming if the input graph is very big (with millions of nodes and hundreds of millions of edges). Our algorithm heavily improves the BFS approach by ``cutting the BFSes'', through an efficient pruning procedure, which allows us to either compute the closeness of a node $v$ or stop the corresponding BFS as soon as we are sure that the closeness of $v$ is not among the $k$ largest ones. The worst-case time complexity of this algorithm is the same as the one of the textbook algorithm: indeed, under reasonable complexity assumptions, this worst-case bound cannot be improved. In particular, it was proved in \cite{Abboud2015} that the complexity of finding the most central vertex is at least the complexity of the APSP: faster combinatorial algorithms for this latter problem would imply unexpected breakthrough in complexity theory \cite{Williams2010}. However, in practice, our algorithm heavily improves the APSP approach, as shown by our experiments. Finally, we will apply our algorithm to an interesting and prototypical case study, that is, the IMDB actor collaboration network: in a little more than half an hour, we have been able to identify the evolution of the set of the 10 most central actors, by analysing snapshots of this network taken every 5 years, starting from 1940 and up to 2014.

\subsection{Related Work}

Closeness is a ``traditional'' definition of centrality, and consequently it was not ``designed with scalability in mind'', as stated in \cite{Kang2011}. Also in \cite{Chen2012}, it is said that closeness centrality can ``identify influential nodes'', but it is ``incapable to be applied in large-scale networks due to the computational complexity''. The simplest solution considered was to define different measures, that might be related to closeness centrality \cite{Kang2011}.

A different line of research has tried to develop more efficient algorithms for this problem. As already said, due to theoretical bounds \cite{Borassi2014b,Abboud2015}, the worst-case complexity cannot be improved in general, unless widely believed conjectures prove to be false. Nevertheless, it is possible to design approximation algorithms: the simplest approach samples the distance between a node $v$ and $l$ other nodes $w$, and return the average of all values $d(v,w)$ found \cite{Eppstein2004}. The time-complexity is $\O(lm)$ to approximate the centrality of all nodes. More refined approximation algorithms are provided in \cite{Cohen2014c,Cohen2014d}, based on the concept of All-Distance Sketch, that is, a procedure that computes the distances between $v$ and $\O(\log n)$ other nodes, chosen in order to provide good estimates of the closeness centrality of $v$. Even if these approximation algorithms work quite well, they are not suited to the ranking of nodes, because the difference between the closeness centrality of two nodes might be very small. Nevertheless, they were used in \cite{Okamoto2008}, where the sampling technique developed in \cite{Eppstein2004} was used to actually compute the top $k$ vertices: the result is not exact, but it is exact with high probability. The authors proved that the time-complexity of their algorithm is $\O(mn^{\frac{2}{3}}\log n)$, under the rather strong assumption that closeness centralities are uniformly distributed between $0$ and $D$, where $D$ is the maximum distance between two nodes (in the worst case, the time-complexity of this algorithm is $\O(mn)$).

Other approaches have tried to develop incremental algorithms, that might be more suited to real-world networks analyses. For instance, in \cite{Lim2011}, the authors develop heuristics to determine the $k$ most central vertices in a varying environment. A different work addressed the problem of updating centralities after edge insertion or deletion \cite{Saryuce2013}: for instance, it is shown that it is possible to update the closeness centrality of $1.2$ million authors in the DBLP-coauthorship network $460$ times faster than recomputing it from scratch.

Finally, some works have tried to exploit properties of real-world networks in order to find more efficient algorithms. In \cite{LeMerrer2014}, the authors develop a heuristic to compute the $k$ most central vertices according to different measures. The basic idea is to identify central nodes according to a simple centrality measure (for instance, degree of nodes), and then to inspect a small set of central nodes according to this measure, hoping it will contain the top $k$ vertices according to the ``complex'' measure. Another approach \cite{Olsen2014} tried to exploit the properties of real-world networks in order to develop exact algorithms with worst-case complexity $\O(mn)$, but performing much better in practice. As far as we know, this is the only exact algorithm that is able to efficiently compute the $k$ most central vertices in networks with up to $1$ million nodes.

However, despite this huge amount of research, the major graph libraries still implement the textbook algorithm: among them, Boost Graph Library \cite{Hagberg2008}, Sagemath \cite{Csardi2006}, igraph \cite{Stein2005}, and NetworkX \cite{Siek2001}. This is due to the fact that the only efficient algorithm published until now for top $k$ closeness centrality is \cite{Olsen2014}, which was published only 1 year ago, and it is quite complicated, because it is based on several other algorithms.

\subsection{Our Results}

As we said before, in this paper we present a new and very simple algorithm for the exact computation of top $k$ closeness central vertices. We show that we drastically improve both the probabilistic approach \cite{Okamoto2008}, and the best algorithm available until now \cite{Olsen2014}. We have computed for the first time the $10$ most central nodes in networks with millions of nodes and hundreds of millions of edges, in very little time. A significant example is the IMDB actor network ($1\,797\,446$ nodes and $72\,880\,156$ edges): we have computed the $10$ most central actors in less than 10 minutes. Moreover, in our DBLP co-authorship network (which should be quite similar to the network used in \cite{Saryuce2013}), our performance is more than $6\,000$ times better than the performance of the textbook algorithm: if only the most central node is needed, we can recompute it from scratch more than $10$ times faster than performing their update. Finally, our approach is not only very efficient, but it is also very easy to code, making it a very good candidate to be implemented in existing graph libraries.

\subsection{Preliminary Definitions}\label{sec:defs}

We assume the reader to be familiar with the basic notions of graph theory (see, for example,~\cite{Cormen2009}): all the notations and definitions used throughout this paper are summarised in \Cref{tbl:notation}. In addition to these definitions, let us precisely define the weighted directed acyclic graph $\mathcal{G=(V,E)}$ of strongly connected components (in short, SCCs) corresponding to a directed graph $G=(V,E)$. In this graph, $\mathcal V$ is the set of SCCs of $G$, and, for any two SCCs $C,D \in \mathcal V$, $(C,D)\in\mathcal E$ if and only if there is an arc in $E$ from a node in $C$ to the a node in $D$. For each SCC $C\in \mathcal V$, the weight $w(C)$ of $C$ is equal to $|C|$, that is, the number of nodes in the SCC $C$. Note that, for each node $v \in C$, $\reac v=\sum_{D\in \Reac C}w(D)$, where $\Reac C$ denotes the set of SCCs that are reachable from $C$ in $\mathcal G$, and $\reac C$ denotes its cardinality.

\subsection{Structure of the Paper}

In \Cref{sec:alg}, we will explain how our algorithm works, and we will prove its correctness. \Cref{sec:exp} will experimentally prove that our algorithm outperforms the best available algorithms, by performing several tests on a dataset of real-world networks. \Cref{sec:imdb} applies the new algorithm to the analysis of the $10$ most central actors in the IMDB actor network.

\section{The Algorithm} \label{sec:alg}
\subsection{Overview} \label{sec:overview}

As we already said, the textbook algorithm for computing the $k$ vertices with largest closeness performs a BFS from each vertex $v$,  computes its closeness $\clos v$, and, finally, returns the $k$ vertices with biggest $\clos v$ values. Similarly, our algorithm (see \Cref{alg:main}) sets $\clos v=\bfsctt(v,x_k)$, where $x_k$ is the $k$-th biggest closeness value found until now ($x_k=0$ if we have not processed at least $k$ vertices). If $\bfsctt(v,x_k)=0$, it means that $v$ is not one of $k$ most central vertices, otherwise $\clos v$ is the actual closeness of $v$. This means that, at the end, the $k$ vertices with biggest closeness values are again the $k$ most central vertices. In order to speed-up the function \bfsctt$(v,x_k)$, we want $x_k$ to be as big as possible, and consequently we need to process central vertices as soon as possible. To this purpose, following the idea of \cite{Lim2011}, we process vertices in decreasing order of degree.

\begin{algorithm}[ht]
\SetKwFunction{bfscf}{\bfsc}
\SetKwFunction{preprocf}{\preproc}
\SetKwFunction{kth}{Kth}
\SetKwFunction{topk}{TopK}
\preprocf($G$); // see Section~\ref{sec:alphaomega}\\
$\clos v \gets 0$ for each $v$; \\
$x_k \gets 0$; \\
\For{$v \in V$ in decreasing order of degree}{
    $\clos v \gets$ \bfscf($v,x_k$); \\
    \If{$\clos v \neq 0$}{
        $x_k \gets \kth(c)$; \\
    }
}
\Return $\topk(c)$;
\caption{overview of the new algorithm. The function \texttt{Kth}$(c)$ returns the $k$-th biggest element of $c$ and the function \texttt{TopK}$(c)$ returns the $k$ biggest elements.}
\label{alg:main}
\end{algorithm}

As we will see in the next sections, \Cref{alg:main} needs linear time in order to execute the preprocessing. Moreover, it requires time $\O(n \log n)$ to sort vertices, $\O(\log k)$ to perform function \texttt{\upshape Kth}, and $\O(1)$ to perform function \texttt{\upshape TopK}, by using a priority queue containing at each step the $k$ most central vertices. Since all other operations need time $\O(1)$, the total running time is $\O(m+n \log n+n\log k+T)=\O(m+n\log n+T))$, where $T$ is the time needed to perform the function \bfsctt\ $n$ times.

Before explaining in details how the function \bfsctt\ operates, let us note that we can easily parallelise the \texttt{for} loop in \Cref{alg:main}, by giving each vertex to a different thread. The parallelisation is almost complete, because the only ``non-parallel'' part deals with the shared variable $x_k$, which must be updated in a synchronised way. In any case, this shared variable does not affect performances. First of all, the time needed for the update is very small, compared to the time of BFSes. Moreover, in principle, a thread might use an ``old'' value of $x_k$, and consequently lose performance: we will experimentally show that the effect of this positive race condition is negligible (\Cref{sec:expbig}).

\subsection{An Upper Bound on the Closeness of Nodes}

The goal of this section is to define an upper bound $\uclos v$ on the closeness of a node $v$, which has to be updated whenever, for any $d \geq 0$, all nodes in $\Gamma_d(v)$ has been reached by the BFS starting from $v$ (that is, whenever the exploration of the $d$-th level of the BFS tree is finished). More precisely, this upper bound is obtained by first proving a lower bound on the farness of $v$, as shown in the following lemma.

\begin{lemma} \label{lem:boundokconn}
For each vertex $v$ and for each $d \geq 0$, 
\[\farn v \geq \lfarn dv{\reac v}\]
(see \Cref{tbl:notation} for the definition of $\lfarn dv{\reac v}$).
\end{lemma}
\begin{proof}
From the definitions of \Cref{tbl:notation}, it follows that 
\[\farn v \geq \dfarn dv + (d+1)\sneig {d+1}v+(d+2)(\reac v-\sball {d+1}v).\]

Since $\sball {d+1}v=\sneig {d+1}v+\sball dv$, 
\[\farn v \geq \dfarn dv -\sneig {d+1}v+(d+2)(\reac v-\sball dv).\]

Finally, since $\uneig {d+1}v$ is an upper bound on $\sneig {d+1}v$,
\[\farn v \geq \dfarn dv -\uneig {d+1}v+(d+2)(\reac v-\sball dv)=\lfarn dv{\reac v}\]
and the lemma follows.\qed
\end{proof}

\begin{corollary}\label{cor:ubound}
For each vertex $v$ and for each $d \geq 0$, $\clos v \leq \uclos dv$, where $\uclos dv$ is defined in \Cref{tbl:notation}.
\end{corollary}

\subsection{Computing the Upper Bound}\label{sec:ubound}

Apart from $\reac v$, all quantities necessary to compute $\lfarn dv{\reac v}$ (and, hence, to compute the upper bound of Lemma~\ref{cor:ubound}) are available as soon as all vertices in $\ball dv$ are visited by a BFS. Note that, if the graph $G$ is (strongly) connected, then $\reac v=n$ is also available. Moreover, if the graph $G$ is undirected (but not necessarily connected), we can compute $\reac v$ for each vertex $v$ in linear time, at the beginning of the algorithm (by simply computing the connected components of $G$).\footnote{Note that if $G$ is undirected, Lemma~\ref{lem:boundokconn} still holds if we redefine, for any $d\geq 1$, $\uneig {d+1}v=\sum_{u \in \neig dv} (\degr u-1)$, because at least one edge from each vertex $u$ in $\neig dv$ connects $u$ to a node in $\neig {d-1}v$.
} It thus remain to deal with the case in which $G$ is directed and not strongly connected.

In this case, let us assume, for now, that we know a lower (respectively, upper) bound $\alpha(v)$ (respectively, $\omega(v)$) on $\reac v$ (see also Table~\ref{tbl:notation}): without loss of generality we can assume that $\alpha(v)>1$. The next lemma shows that, instead of examining all possible values of $\reac v$ between $\alpha(v)$ and $\omega(v)$, it is sufficient to examine only the two extremes of this interval. In order to apply this idea, however, we have to work with the inverse $\frac{1}{\clos v}$ of the closeness of $v$, because otherwise the denominator might vanish or even be negative (due to the fact that we are both upper bounding $\sneig {d+1}v$ and lower bounding $\reac v$). In other words, the lemma will provide us a lower bound $\lambda_d(v)$ on $\frac{1}{\clos v}$ (so that, if $\lambda_d(v)$ is negative, then $\lambda_d(v) \leq \frac{1}{\clos v}$ trivially holds).

\begin{lemma}\label{lem:lowbound}
For each vertex $v$ and for each $d\geq 0$,$$\frac{1}{\clos v} \geq \lambda_d(v):=(n-1)\min\left(\frac{\lfarn dv{\alpha(v)}}{(\alpha(v)-1)^2}, \frac{\lfarn dv{\omega(v)}}{(\omega(v)-1)^2}\right).$$
\end{lemma}
\begin{proof}
From \Cref{lem:boundokconn}, if we denote $a=d+2$ and $b=\uneig {d+1}v+a(\sball dv-1)-\dfarn dv$,
\begin{align*}
\farn v &\geq \dfarn dv-\uneig {d+1}v+a({\reac v}-\sball dv) \\
&= a({\reac v}-1)+\dfarn dv-\uneig {d+1}v-a(\sball dv-1) \\
&= a({\reac v}-1)-b.
\end{align*}
Note that $a>0$ because $d>0$, and $b>0$ because
\[\dfarn dv = \sum_{w \in \ball dv} d(v,w) \leq d(\sball dv-1) < a(\sball dv-1)\]
where the first inequality holds because, if $w=v$, then $d(v,w)=0$, and if $w\in \ball dv$, then $d(v,w)\leq d$. 
Hence, $\frac{1}{c_v} \geq (n-1)\frac{a(\reac v-1)-b}{(\reac v -1)^2}$. Let us consider the function $g(x)=\frac{ax-b}{x^2}$. The derivative $g'(x)=\frac{-ax+2b}{x^3}$ is positive for $0<x<\frac{2b}{a}$ and negative for $x>\frac{2b}{a}$: this means that $\frac{2b}{a}$ is a local maximum, and there are no local minima for $x>0$. Consequently, in each closed interval $[x_1,x_2]$ where $x_1$ and $x_2$ are positive, the minimum of $g(x)$ is reached in $x_1$ or $x_2$. Since $0<\alpha(v)-1 \leq \reac v-1 \leq \omega(v)-1$, 
\[g(\reac v-1) \geq \min(g(\alpha(v)-1), g(\omega(v)-1))\]
and the conclusion follows.
\end{proof}

\subsection{Computing $\alpha(v)$ and $\omega(v)$}\label{sec:alphaomega}

It now remains to compute $\alpha(v)$ and $\beta(v)$ (in the case of a directed graph which is not strongly connected). This can be done during the preprocessing phase of our algorithm as follows. Let $\mathcal{G=(V,E)}$ be the weighted directed acyclic graph of SCCs, as defined in Section~\ref{sec:defs}. We already observed that, if $v$ and $w$ are in the same SCC, then $\reac v=\reac w=\sum_{D\in \mathcal R(C)}w(D)$, where $\mathcal R(C)$ denotes the set of SCCs that are reachable from $C$ in $\mathcal G$. This means that we simply need to compute a lower (respectively, upper) bound $\alpha(C)$ (respectively, $\omega(C)$) on $\sum_{D\in \mathcal R(C)}w(D)$, for every SCC $C$. To this aim, we first compute a topological sort $\{C_1, \dots, C_l\}$ of $\mathcal V$ (that is, if $(C_i,C_j) \in \mathcal E$, then $i<j$). Successively, we use a dynamic programming approach, and, by starting from $C_l$, we process the SCCs in reverse topological order, and we set$$\alpha(C)=w(C) + \max_{(C, D) \in \mathcal E} \alpha(D)$$and$$\omega(C)=w(C) + \sum_{(C, D) \in \mathcal E} \omega(D).$$ Note that processing the SCCs in reverse topological ordering ensures that the values $\alpha(D)$ and $\omega(D)$ on the right hand side of these equalities are available when we process the SCC $C$. Clearly, the complexity of computing $\alpha(C)$ and $\omega(C)$, for each SCC $C$, is linear in the size of $\mathcal G$, which in turn is smaller than $G$.

Observe that the bounds obtained through this simple approach can be improved by using some ``tricks''. First of all, when the biggest SCC $\tilde{C}$ is processed, we do not use the dynamic programming approach and we can exactly compute $\sum_{D\in \mathcal R(\tilde{C})}w(D)$ by simply performing a BFS starting from any node in $\tilde{C}$. This way, not only $\alpha(\tilde{C})$ and $\omega(\tilde{C})$ are exact, but also $\alpha(C)$ and $\omega(C)$ are improved for each SCC $C$ from which it is possible to reach $\tilde{C}$. Finally, in order to compute the upper bounds for the SCCs that are able to reach $\tilde{C}$, we can run the dynamic programming algorithm on the graph obtained from $\mathcal{G}$ by removing all components reachable from $\tilde{C}$, and we can then add $\sum_{D\in \mathcal R(\tilde{C})}w(D)$. 

\subsection{The \bfsc\ Function} \label{sec:bfsc}

We are now ready to define the function \bfsctt$(v,x)$, which returns the closeness $\clos v$ of vertex $v$ if $\clos v>x$, $0$ otherwise. To do so, this function performs a BFS starting from $v$, and during the BFS it updates the upper bound $\uclos dv \geq \clos v$ (the update is done whenever all nodes in $\Gamma_d(v)$ has been reached): as soon as $\uclos dv \leq x$, we know that $\clos v \leq \uclos dv \leq x$, and we return $0$. If this situation never occurs, at the end of the visit we have clearly computed $\clos v$.

Algorithm~\ref{alg:bfscut} is the pseudo-code of the function \bfsctt\ when implemented for strongly connected graphs (recall that, in this case, $r(v) = n$): this code can be easily adapted to the case of undirected graphs (see the beginning of Section~\ref{sec:ubound}) and to the case of directed (not necessarily strongly connected) graphs (see Lemma~\ref{lem:lowbound}).

\begin{algorithm}[ht]
\SetKwFunction{bfscf}{\bfsc}
\SetKwFunction{preprocf}{\preproc}
\SetKwFunction{kth}{Kth}
\SetKwFunction{topk}{TopK}
Create queue $Q$;\\
$Q$.enqueue($v$);\\
Mark $v$ as visited;\\
$d \gets 0$; $f \gets 0$;  $\tilde{\gamma} \gets 0$; $nd \gets 0$; \\
\While{$Q$ is not empty}{
    $u \gets Q$.dequeue(); \\
    \If{$d(v,u)>d$}{
        $\tilde{f} \gets f-\tilde{\gamma}+(d+2)(n-nd)$;\\
        $\tilde{c} \gets \frac{(n-1)}{\tilde{f}}$;\\
        \If{$\tilde{c}\leq x$}{\Return 0;}
        $d \gets d+1$;\\
    }
    $f \gets f+d(v,u)$;\\
    $\tilde{\gamma} \gets \tilde{\gamma}+\odegr u$;\\
    $nd \gets nd+1$;\\
    \For{$w$ in adjacency list of $u$}{
        \If{$w$ is not visited}{
            $Q$.enqueue($w$);\\
            Mark $w$ as visited;\\
        }
    }
}
\Return $\frac{n-1}{f}$;
\caption{the \bfsctt$(v,x)$ function in the case of strongly connected directed graphs.}
\label{alg:bfscut}
\end{algorithm}

\section{Experimental Results} \label{sec:exp}
In this section, we will test our algorithm on several real-world networks, in order to show its performances. All the networks used in our experiments were collected from the datasets SNAP (\url{snap.stanford.edu/}), NEXUS  (\url{nexus.igraph.org}), LASAGNE (\url{piluc.dsi.unifi.it/lasagne}), LAW (\url{law.di.unimi.it}), and IMDB (\url{www.imdb.com}). Our tests have been performed on a server running an Intel(R) Xeon(R) CPU E5-4607 0, 2.20GHz, with 48 cores, 250GB RAM, running Ubuntu 14.04\_2 LTS; our code has been written in Java 1.7, and it is available at \url{tinyurl.com/kelujv7}.

\subsection{Comparison with the State of the Art} \label{sec:expsmall}

In order to compare the performance of our algorithm with state of the art approaches, we have selected 21 networks whose number of nodes ranges between $1\,000$ and $30\,000$ nodes.

We have compared our algorithm \newalg\ with our implementations of the best available algorithms for top-$k$ closeness centrality.\footnote{Note that the source code of our competitors is not available.} The first one \cite{Olsen2014} is based on a pruning technique and on $\Delta$-BFS, a method to reuse information collected during a BFS from a node to speed up a BFS from one of its in-neighbors; we will denote this algorithm as \deltabfs. The second one provides top-$k$ closeness centralities with high probability \cite{Okamoto2008}. It is based on performing BFSes from a random sample of nodes to estimate the closeness centrality of all the other nodes, then computing the exact centrality of all the nodes whose estimate is big enough. Note that this algorithm requires the input graph to be (strongly) connected: for this reason, differently from the other algorithms, we have run this algorithm on the largest (strongly) connected component. We will denote this latter algorithm as \sampl.

In order to perform a fair comparison we have considered the \emph{improvement factor} $\frac{\mvis}{\mtot}$, where $\mvis$ is the number of arcs visited during the algorithm and $\mtot$ is the number of arcs visited by the \emph{textbook} algorithm, which is based on all-pair shortest-path. It is worth observing that $\mtot$ is $m\cdot n$ if the graph is (strongly) connected, but it might be smaller in general. The improvement factor does not consider the preprocessing time, which is negligible, and it is closely related to the running-time, since all the three algorithms are based on performing BFSes. Furthermore, it does not depend on the implementation and on the machine used for the algorithm. In the particular case of \deltabfs, we have just counted the arcs visited in BFS and $\Delta$-BFS, without considering all the operations done in the pruning phases (see \cite{Olsen2014}).

Our results are summarized in \Cref{tab:comparison}, where we report the arithmetic mean and the geometric mean of all the improvement factors (among all the graphs). In our opinion, the geometric mean is more significant, because the arithmetic mean is highly influenced by the maximum value. 
\begin{table}[htb]
\caption{comparison of the improvement factors of the new algorithm (\newalg), of the algorithm in \cite{Olsen2014} ({\scshape \deltabfs}), and the algorithm in \cite{Okamoto2008} (\sampl). Values are computed for $k=1$ and $k=10$, and they are averaged over all graphs in the dataset.}
\label{tab:comparison}
\centering
\begin{scriptsize}
\begin{tabular}{|@{\hspace{3pt}}c@{\hspace{3pt}}|l|r|r|r|r|r|r|}
\hline
 & & \multicolumn{3}{c|}{\textsc{Arithmetic Mean}} & \multicolumn{3}{c|}{\textsc{Geometric Mean}} \\
$k$ & \textsc{Alg} &  \multicolumn{1}{c|}{\textsc{Dir}} & \multicolumn{1}{c|}{\textsc{Undir}} & \multicolumn{1}{c|}{\textsc{Both}} & \multicolumn{1}{c|}{\textsc{Dir}} & \multicolumn{1}{c|}{\textsc{Undir}} & \multicolumn{1}{c|}{\textsc{Both}} \\
\hline
& \newalg & \textbf{4.5\%} & \textbf{2.8\%} & \textbf{3.46\%} & \textbf{1.3\%} & \textbf{0.8\%} & \textbf{0.9\%} \\
1 & \deltabfs & 43.5\% & 24.2\% & 31.6\% & 35.6\% & 15.8\% & 21.5\% \\
 & \sampl & 72.3\% & 45.4\% & 55.6\% & 67.3\% & 42.8\% & 50.8\% \\
 \hline
& \newalg & \textbf{14.1\%} & \textbf{5.3\%} & \textbf{8.6\%} & \textbf{6.3\%} & \textbf{2.7\%} & \textbf{3.8\%} \\
10 & \deltabfs & 43.6\% & 24.2\% & 31.6\% & 35.6\% & 15.8\% & 21.6\% \\
 & \sampl & 80.7\% & 59.3\% & 67.5\% & 78.4\% & 57.9\% & 65.0\% \\
\hline
\end{tabular}
\end{scriptsize}
\end{table}

More detailed results are available in \Cref{tab:compext}.

\begin{table*}[tb]
\caption{detailed comparison of the improvement factor of the three algorithms with respect to the all-pair-shortest-path algorithm, with $k=1$ and $k=10$.}
\label{tab:compext}
\centering
\begin{scriptsize}
\begin{tabular}{|l|r|r|r|r|r|r|r|r|}
\multicolumn{9}{c}{\textbf{Directed Networks}} \\
\hline
     & &  &  \multicolumn{3}{c|}{$k=1$} & \multicolumn{3}{c|}{$k=10$} \\
\textsc{Network} & \textsc{Nodes} & \textsc{Edges} & \multicolumn{1}{c|}{\newalg} & \multicolumn{1}{c|}{\deltabfs} & \multicolumn{1}{c|}{\sampl} & \multicolumn{1}{c|}{\newalg} & \multicolumn{1}{c|}{\deltabfs} & \multicolumn{1}{c|}{\sampl} \\
\hline
 polblogs                          &     1224     &     19022     &     3.052\%     &     41.131\%     &     88.323\%     &     8.491\%     &     41.321\%     &     91.992\%\\ 
 p2p-Gnutella08           &     6301     &     20777     &     4.592\%     &     53.535\%     &     87.229\%     &     23.646\%     &     53.626\%     &     92.350\%\\ 
 wiki-Vote                 &     7115     &     103689     &     0.068\%     &     25.205\%     &     40.069\%     &     0.825\%     &     25.226\%     &     62.262\%\\ 
 p2p-Gnutella09           &     8114     &     26013     &     7.458\%     &     55.754\%     &     86.867\%     &     18.649\%     &     55.940\%     &     90.248\%\\ 
 p2p-Gnutella06           &     8717     &     31525     &     0.808\%     &     52.615\%     &     77.768\%     &     18.432\%     &     52.831\%     &     88.884\%\\ 
 freeassoc                         &     10617     &     72172     &     17.315\%     &     58.204\%     &     85.831\%     &     20.640\%     &     57.954\%     &     87.300\%\\ 
 p2p-Gnutella04           &     10876     &     39994     &     2.575\%     &     56.788\%     &     84.128\%     &     21.754\%     &     56.813\%     &     89.961\%\\ 
 as-caida20071105             &     26475     &     106762     &     0.036\%     &     4.740\%     &     27.985\%     &     0.100\%     &     4.740\%     &     42.955\%\\ 
\hline
\multicolumn{9}{c}{}\\
\multicolumn{9}{c}{\textbf{Undirected Networks}} \\
\hline
     & &  &  \multicolumn{3}{c|}{$k=1$} & \multicolumn{3}{c|}{$k=10$} \\
\textsc{Network} & \textsc{Nodes} & \textsc{Edges} & \multicolumn{1}{c|}{\newalg} & \multicolumn{1}{c|}{\deltabfs} & \multicolumn{1}{c|}{\sampl} & \multicolumn{1}{c|}{\newalg} & \multicolumn{1}{c|}{\deltabfs} & \multicolumn{1}{c|}{\sampl} \\
\hline
Homo                    &     1027     &     1166     &     5.259\%     &     82.794\%     &     82.956\%     &     14.121\%     &     82.794\%     &     88.076\%\\ 
HC-BIOGRID              &     4039     &     10321     &     5.914\%     &     19.112\%     &     65.672\%     &     8.928\%     &     19.112\%     &     72.070\%\\ 
Mus\_musculus            &     4610     &     5747     &     1.352\%     &     7.535\%     &     55.004\%     &     5.135\%     &     7.535\%     &     66.507\%\\ 
Caenorhabditis\_elegans     &     4723     &     9842     &     1.161\%     &     9.489\%     &     45.623\%     &     1.749\%     &     9.489\%     &     58.521\%\\ 
 ca-GrQc                      &     5242     &     14484     &     3.472\%     &     13.815\%     &     55.099\%     &     5.115\%     &     13.815\%     &     62.523\%\\ 
advogato               &     7418     &     42892     &     0.427\%     &     82.757\%     &     41.364\%     &     0.891\%     &     82.757\%     &     61.688\%\\ 
hprd\_pp                 &     9465     &     37039     &     0.219\%     &     15.827\%     &     44.084\%     &     2.079\%     &     15.827\%     &     54.300\%\\ 
 ca-HepTh                    &     9877     &     25973     &     2.796\%     &     15.474\%     &     46.257\%     &     3.630\%     &     15.474\%     &     52.196\%\\ 
 Drosophila\_melanogaster      &     10625     &     40781     &     1.454\%     &     18.347\%     &     40.513\%     &     1.991\%     &     18.347\%     &     46.847\%\\ 
 oregon1\_010526              &     11174     &     23409     &     0.058\%     &     4.937\%     &     28.221\%     &     0.233\%     &     4.937\%     &     49.966\%\\ 
 oregon2\_010526               &     11461     &     32730     &     0.090\%     &     5.848\%     &     23.780\%     &     0.269\%     &     5.848\%     &     40.102\%\\ 
GoogleNw               &     15763     &     148585     &     0.007\%     &     7.377\%     &     33.501\%     &     4.438\%     &     7.377\%     &     75.516\%\\ 
 dip20090126\_MAX        &     19928     &     41202     &     14.610\%     &     31.627\%     &     27.727\%     &     20.097\%     &     31.673\%     &     42.901\% \\
 \hline
\end{tabular}
\end{scriptsize}
\end{table*}

In the case $k=1$ (respectively, $k=10$), the geometric mean of the improvement factor of \newalg\ is $23$ (resp.~$6$) times smaller than \deltabfs\ and $54$ times smaller than \sampl\ (resp.~$17$). Moreover we highlight that the new algorithm outperforms all the competitors in each single graph, both with $k=1$ and with $k=10$.

We have also tested our algorithm on the three unweighted graphs analyzed in \cite{Olsen2014}, respectively called \texttt{Web}, \texttt{Wiki}, and \texttt{DBLP}.  
By using a single thread implementation of \newalg, in the \texttt{Web} graph (resp. \texttt{DBLP}) we computed the top-10 nodes in 10 minutes (resp. 10 minutes) on the whole graph, having 875\,713 nodes (resp. 1\,305\,444), while \deltabfs\ needed about 25 minutes (resp. 4 hours) for a subgraph of 400\,000 nodes. The most striking result deals with \texttt{Wiki}, where \newalg\ needed 30 seconds for the whole graph having 2\,394\,385 nodes instead of about 15 minutes on a subgraph with 1 million nodes. Using multiple threads our  performances are even better, as we will show in \Cref{sec:parallel}.

\subsection{Real-World Large Networks} \label{sec:expbig}

In this section, we will run our algorithm on a bigger dataset, composed by 25 directed and 15 undirected networks, with up to 7\,414\,768 nodes and 191\,606\,827 edges.
Once again, we will consider the number of visited arcs by \newalg, i.e. $\mvis$, but this time we will analyze the performance ratio $\frac{\mvis}{mn}$ instead of the improvement factor. Indeed, due to the large size of these networks, the textbook algorithm did not finish in a reasonable time.

It is worth observing that we have been able to compute for the first time the $k$ most central nodes of networks with millions of nodes and hundreds of millions of arcs, with $k=1$ and $k=10$. The detailed results are shown in \Cref{tab:exp}, where for each network we have reported the performance ratio, both for $k=1$ and $k=10$. A summary of these results is provided by \Cref{tab:mean}, providing the arithmetic and geometric means.

\begin{table}[t]
\caption{the arithmetic and geometric mean of the performance ratios (percentage).}
\label{tab:mean}
\centering
\begin{scriptsize}
\begin{tabular}{|c|r|r|r|r|r|r|}
\hline
 & \multicolumn{3}{c|}{\textsc{Arithmetic Mean}} & \multicolumn{3}{c|}{\textsc{Geometric Mean}} \\
$k$ & \multicolumn{1}{c|}{\textsc{Dir}} & \multicolumn{1}{c|}{\textsc{Undir}} & \multicolumn{1}{c|}{\textsc{Both}} & \multicolumn{1}{c|}{\textsc{Dir}} & \multicolumn{1}{c|}{\textsc{Undir}} & \multicolumn{1}{c|}{\textsc{Both}} \\
\hline
1 & 2.89\% & 1.14\% & 2.24\%  & 0.36\% & 0.12\%  & 0.24\% \\
\hline
10 & 3.82\% & 1.83\% & 3.07\%  & 0.89\% & 0.84\% & 0.87\% \\
\hline
\end{tabular}
\end{scriptsize}
\end{table}

\begin{table*}[t]
\caption{performance ratio of the new algorithm.}
\label{tab:exp}
\centering
\begin{scriptsize}
\begin{tabular}{|l|r|r|r|r|}
\multicolumn{5}{c}{} \\
\multicolumn{5}{c}{\textbf{Directed Networks}} \\
\hline
 &  &  & \multicolumn{2}{c|}{\textsc{Performance Ratio (\%)}} \\
\textsc{Network} & \multicolumn{1}{c|}{\textsc{Nodes}} & \multicolumn{1}{c|}{\textsc{Edges}} & \multicolumn{1}{c|}{$k=1$} & \multicolumn{1}{c|}{$k=10$} \\
\hline
%wiki-Vote & 7115 & 103689 & 0.02751 & 0.33265\\
%as-caida20071105 & 26475 & 106762 & 0.03559 & 0.09869\\
cit-HepTh & 27770 & 352768 & 2.31996 & 4.65612\\
cit-HepPh & 34546 & 421534 & 1.21083 & 1.64227\\
p2p-Gnutella31 & 62586 & 147892 & 0.73753 & 2.51074\\
soc-Epinions1 & 75879 & 508837 & 0.16844 & 1.92346\\
soc-sign-Slashdot081106 & 77350 & 516575 & 0.59535 & 0.65012\\
soc-Slashdot0811 & 77360 & 828161 & 0.01627 & 0.48842\\
twitter-combined & 81306 & 1768135 & 0.76594 & 1.03332\\
soc-sign-Slashdot090216 & 81867 & 545671 & 0.5375 & 0.58774\\
soc-sign-Slashdot090221 & 82140 & 549202 & 0.54833 & 0.5995\\
soc-Slashdot0902 & 82168 & 870161 & 0.01662 & 0.76048\\
%uk-2007-05@100000 & 100000 & 3034945 & 0.20819 & 0.02514\\
gplus-combined & 107614 & 13673453 & 0.36511 & 0.3896\\
amazon0302 & 262111 & 1234877 & 10.16028 & 11.90729\\
email-EuAll & 265214 & 418956 & 0.00192 & 0.00764\\
web-Stanford & 281903 & 2312497 & 6.55454 & 10.67736\\
%cnr-2000 & 325557 & 3128710 & 0.09612 & 0.04424\\
web-NotreDame & 325729 & 1469679 & 0.04592 & 0.62945\\
amazon0312 & 400727 & 3200440 & 8.17931 & 9.40879\\
amazon0601 & 403394 & 3387388 & 7.80459 & 9.33853\\
amazon0505 & 410236 & 3356824 & 7.81823 & 9.11571\\
web-BerkStan & 685230 & 7600595 & 21.0286 & 23.70427\\
%eu-2005 & 862664 & 18733713 & 0.12895 & 0.27151\\
web-Google & 875713 & 5105039 & 0.13662 & 0.23239\\
%uk-2007-05@1000000 & 1000000 & 41042835 & 0.17133 & 0.02066\\
in-2004 & 1382870 & 16539643 & 1.81649 & 2.51671\\
soc-pokec-relationships & 1632803 & 30622564 & 0.00411 & 0.02257\\
wiki-Talk & 2394385 & 5021410 & 0.00029 & 0.00247\\
indochina-2004 & 7414768 & 191606827 & 1.341 & 2.662 \\
 \hline
 
 \multicolumn{5}{c}{} \\
\multicolumn{5}{c}{\textbf{Undirected Networks}} \\
\hline
 &  &  & \multicolumn{2}{c|}{\textsc{Performance Ratio (\%)}} \\
\textsc{Network} & \multicolumn{1}{c|}{\textsc{Nodes}} & \multicolumn{1}{c|}{\textsc{Edges}} & \multicolumn{1}{c|}{$k=1$} & \multicolumn{1}{c|}{$k=10$} \\
\hline
%ca-GrQc & 5242 & 14484 & 2.55692 & 3.75991\\
%ca-HepTh & 9877 & 25973 & 2.33564 & 3.03352\\
%oregon1-010526 & 11174 & 23409 & 0.05811 & 0.23331\\
%oregon2-010526 & 11461 & 32730 & 0.08996 & 0.26937\\
ca-HepPh & 12008 & 118489 & 9.41901 & 9.57862\\
ca-AstroPh & 18772 & 198050 & 1.35832 & 3.2326\\
ca-CondMat & 23133 & 93439 & 0.23165 & 1.07725\\
dblp-conf2015-net-bigcomp & 31951 & 95084 & 2.46188 & 3.42476\\
email-Enron & 36692 & 183831 & 0.10452 & 0.28912\\
loc-gowalla-edges & 196591 & 950327 & 0.00342 & 3.04066\\
com-dblp.ungraph & 317080 & 1049866 & 0.20647 & 0.312\\
com-amazon.ungraph & 334863 & 925872 & 2.55046 & 3.08037\\
com-youtube.ungraph & 1134890 & 2987624 & 0.04487 & 0.60811\\
dblp22015-net-bigcomp & 1305444 & 6108712 & 0.01618 & 0.0542\\
as-skitter & 1696415 & 11095298 & 0.54523 & 0.61078\\
com-orkut.ungraph & 3072441 & 117185083 & 0.00241 & 0.38956 \\
\hline
\end{tabular}
\end{scriptsize}
\end{table*}

First of all, we note that the values obtained are impressive: the geometric mean is always below $1\%$, and for $k=1$ it is even smaller. The arithmetic mean is slightly bigger, mainly because of \texttt{amazon} product-co-purchasing networks, two web networks and one collaboration network, where the performance ratio is quite high. 
Most of the other networks have a very low performance ratio: with $k=1$, $65\%$ of the networks are below $1\%$, and $32.5\%$ of the networks are below $0.1\%$. With $k=10$, $52.5\%$ of the networks are below $1\%$ and $12.5\%$ are below $0.1\%$.

We also outline that in some cases the performance ratio is even smaller: a striking example is \texttt{com-Orkut}, where our algorithm for $k=1$ is more than $40\,000$ times faster than the \emph{textbook} algorithm, whose performance is $m\cdot n$, because the graph is connected.

%\todo{Add plot with nodes vs performance}

\subsection{Multi-Thread Experiments} \label{sec:parallel}

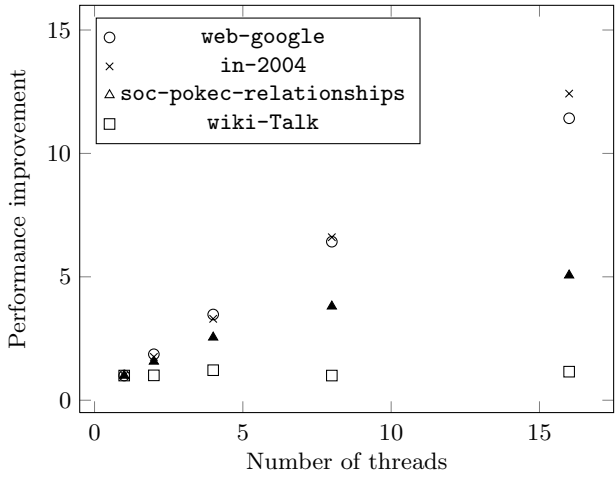
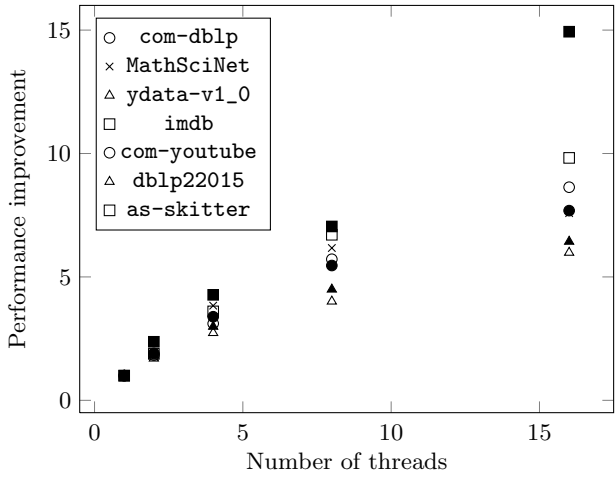
\begin{figure}[t]
\centering
\begin{subfigure}[t]{0.49\textwidth}
\begin{tikzpicture}
\tikzset{
did/.style={rectangle,draw=none,fill=none,inner sep=0cm},
every node/.style={did}}

\pgfplotsset{cycle list={{mark=*}},height=7cm,width=\textwidth}
\begin{axis}[xlabel=Number of threads,
			ylabel=Performance improvement,
			y label style={did,at={(0.06,0.5)}},
			x tick label style={did, minimum height = 0.5cm},
			ymax=16,
			x label style={did},
			y tick label style={did, inner sep=3pt},
			legend pos=north west,
			legend entries={\texttt{web-google},\texttt{in-2004},\texttt{soc-pokec-relationships},\texttt{wiki-Talk}}]
			                   
\addplot+[mark=o,only marks] coordinates {
(1,1)
(2,1.8625782969)
(4,3.47387235)
(8,6.4229301753)
(16,11.4237030445)
};
\addplot+[mark=x,only marks] coordinates {
(1,1)
(2,1.755299418)
(4,3.2991375955)
(8,6.6038075903)
(16,12.4233045487)
};
\addplot+[mark=triangle*,only marks] coordinates {
(1,1)
(2,1.5699385627)
(4,2.5461707797)
(8,3.8033921772)
(16,5.0614078019)
};

\addplot+[mark=square,only marks] coordinates {
(1,1)
(2,1.0082883577)
(4,1.2167946034)
(8,0.9996178101)
(16,1.1567890314)
};

\end{axis}
\end{tikzpicture}

\caption{Directed networks.}
\end{subfigure} \\

\vspace{0.5cm}

\begin{subfigure}[t]{0.49\textwidth}
\begin{tikzpicture}[]
\tikzset{
did/.style={rectangle,draw=none,fill=none,inner sep=0cm},
every node/.style={did}}

\pgfplotsset{cycle list={{mark=*}},height=7cm,width=\textwidth}
\begin{axis}[xlabel=Number of threads,
			ylabel=Performance improvement,
			y label style={did,at={(0.06,0.5)}},
			x tick label style={did, minimum height = 0.5cm},
			x label style={did},
			y tick label style={did, inner sep=3pt},
			ymax=16,
			legend pos=north west,
			legend entries={\texttt{com-dblp},\texttt{MathSciNet},\texttt{ydata-v1\_0},\texttt{imdb},\texttt{com-youtube},\texttt{dblp22015},\texttt{as-skitter}}]
			
%\draw[ultra thin] (axis cs:\pgfkeysvalueof{/pgfplots/xmin},1) -- (axis cs:\pgfkeysvalueof{/pgfplots/xmax},1);

\addplot+[mark=o,only marks] coordinates {
(1,1)
(2,1.8319216026)
(4,3.1097292191)
(8,5.7161184528)
(16,8.6345621448)
};
\addplot+[mark=x,only marks] coordinates {
(1,1)
(2,1.9249589298)
(4,3.8359639892)
(8,6.1717456734)
(16,7.5841423948)
};
\addplot+[mark=triangle,only marks] coordinates {
(1,1)
(2,1.6983666626)
(4,2.7329027554)
(8,4.0081432478)
(16,5.9853756474)
};

\addplot+[mark=square,only marks] coordinates {
(1,1)
(2,1.881716489)
(4,3.6038475269)
(8,6.7080545932)
(16,9.8218092375)
};
\addplot+[mark=*,only marks] coordinates {
(1,1)
(2,1.9010187924)
(4,3.3872238488)
(8,5.4619511567)
(16,7.6855900471)
};
\addplot+[mark=triangle*,only marks] coordinates {
(1,1)
(2,1.8267242053)
(4,2.9880411317)
(8,4.4937095499)
(16,6.4315367919)
};
\addplot+[mark=square*,only marks] coordinates {
(1,1)
(2,2.3721484246)
(4,4.2745880838)
(8,7.04450837)
(16,14.9346062401)
};
\end{axis}
\end{tikzpicture}
\caption{Undirected networks.}
\end{subfigure}
\caption{the benefits of parallelization.}
\label{fig:parallel}
\end{figure}

In this section, we will test the performance of a parallel version of \newalg\ (see \Cref{sec:overview}). In particular, we have considered the ratio between the time needed to compute the most central node with one thread and with $x$ threads, where $x\in\{1,2,4,8,16\}$. This ratio is plotted in \Cref{fig:parallel} for $k=1$ (for $k=10$ very similar results hold). Ideally, the ratio should be very close to the number of threads; however, due to memory access, the actual ratio is smaller. For graphs where the performance ratio is small, like in the case of \texttt{wiki-Talk} (see \Cref{tab:exp} in the appendix), the running time is mostly consumed by the preprocessing phase (even if it is $\O(n \log n)$); in these cases, we observe that there seems to be no room for parallelization improvements. On the other hand, when the computation is more time consuming, like in the case of \texttt{web-google} or \texttt{as-skitter}, the parallelization is very efficient and close to the optimum.

We have also tested if the positive race condition in the update of $x_k$ affects performances (see \Cref{sec:overview}), by considering the performance ratio with a varying number of threads. In \Cref{tab:posrace}, we report how much the performance ratio increases if the algorithm is run with $16$ threads instead of $1$ (more formally, if $p_i$ is the performance ratio with $i$ threads, we report $\frac{p_{16}-p_1}{p_1}$). For instance, a value of $100\%$ means that the number of visited arcs has doubled: ideally, this should result in a factor $8$ speedup, instead of a factor $16$ (that is, the number of threads).
\begin{table}[t]
\caption{the increase in performance ratio from $1$ to $16$ threads.}
\label{tab:posrace}
\centering
\begin{tabular}{|l|r|r|}
\hline
\textsc{Network} & $k=1$ & $k=10$ \\
\hline
 MathSciNet                 &  3.88\%  &  0.18\%\\ 
 com-dblp           &  1.32\%  &  0.74\%\\ 
 ydata-v1-0                 &  1.40\%  &  0.24\%\\ 
 wiki-Talk                  &  192.84\%  &  49.40\%\\ 
 web-Google                 &  1.70\%  &  0.95\%\\ 
 com-youtube        &  1.91\%  &  0.10\%\\ 
 dblp22015     &  2.54\%  &  0.47\%\\ 
 as-skitter                 &  0.14\%  &  0.18\%\\ 
 in-2004                    &  0.09\%  &  0.10\% \\ 
 soc-pokec-relationships    &  17.71\%  & 4.51\% \\ 
 imdb                       &  5.61\%  & 4.03\% \\
 \hline
\end{tabular}
\end{table}
We observe that these values are very small, especially for $k=10$, where all values except one are below $5\%$. This means that, ideally, instead of a factor $16$ improvement, we obtain a factor $15.24$. The only case where this is not verified is \texttt{wiki-Talk}, where in any case the performance ratio is very small (see \Cref{tab:exp}).

\section{IMDB Case Study} \label{sec:imdb}
In this section, we will apply the new algorithm \newalg\ to analyze the IMDB graph, where nodes are actors, and two actors are connected if they played together in a movie (TV-series are ignored). The data collected come from the website \url{http://www.imdb.com}: in line with \url{http://oracleofbacon.org}, we decided to exclude some genres from our database: awards-shows, documentaries, game-shows, news, realities and talk-shows. We analyzed snapshots of the actor graph, taken every 5 years from 1940 to 2010, and 2014. The total time needed to perform the computation was 37 minutes with 30 threads, and the results are reported in \Cref{tab:imdb}.

\begin{small}
\begin{table*}[b]
\caption{detailed results of the IMDB actor graph.}
\label{tab:imdb}
\centering
\begin{tabular}{|l|c|c|c|c|}
\hline
\textsc{Year} & \textbf{1940} & \textbf{1945} & \textbf{1950} & \textbf{1955}\\ 
\textsc{Nodes} & 69\,011 & 83\,068 & 97\,824 & 120\,430\\ 
\textsc{Edges} & 3\,417\,144 & 5\,160\,584 & 6\,793\,184 & 8\,674\,159\\ 
\textsc{Perf.~ratio} & 5.62\% & 4.43\% & 4.03\% & 2.91\%\\ 
\hline
\textsc{1st} &           Semels, Harry (I)  &               Corrado, Gino  &               Flowers, Bess  &               Flowers, Bess \\ 
\textsc{2nd} &               Corrado, Gino  &               Steers, Larry  &               Steers, Larry  &            Harris, Sam (II) \\ 
\textsc{3rd} &               Steers, Larry  &               Flowers, Bess  &               Corrado, Gino  &               Steers, Larry \\ 
\textsc{4th} &              Bracey, Sidney  &           Semels, Harry (I)  &            Harris, Sam (II)  &               Corrado, Gino \\ 
\textsc{5th} &              Lucas, Wilfred  &              White, Leo (I)  &           Semels, Harry (I)  &          Miller, Harold (I) \\ 
\textsc{6th} &              White, Leo (I)  &            Mortimer, Edmund  &           Davis, George (I)  &            Farnum, Franklyn \\ 
\textsc{7th} &           Martell, Alphonse  &               Boteler, Wade  &             Magrill, George  &             Magrill, George \\ 
\textsc{8th} &           Conti, Albert (I)  &             Phelps, Lee (I)  &             Phelps, Lee (I)  &               Conaty, James \\ 
\textsc{9th} &               Flowers, Bess  &                 Ring, Cyril  &                 Ring, Cyril  &           Davis, George (I) \\ 
\textsc{10th} &                 Sedan, Rolfe  &               Bracey, Sidney  &              Moorhouse, Bert  &               Cording, Harry \\ 
\hline
\multicolumn{5}{c}{} \\ 
\hline
\textsc{Year} & \textbf{1960} & \textbf{1965} & \textbf{1970} & \textbf{1975}\\ 
\textsc{Nodes} & 146\,253 & 174\,826 & 210\,527 & 257\,896\\ 
\textsc{Edges} & 11\,197\,509 & 12\,649\,114 & 14\,209\,908 & 16\,080\,065\\ 
\textsc{Perf.~ratio} & 2.21\% & 1.60\% & 1.14\% & 0.83\%\\ 
\hline
\textsc{1st} &               Flowers, Bess  &               Flowers, Bess  &               Flowers, Bess  &               Flowers, Bess \\ 
\textsc{2nd} &            Harris, Sam (II)  &            Harris, Sam (II)  &            Harris, Sam (II)  &            Harris, Sam (II) \\ 
\textsc{3rd} &            Farnum, Franklyn  &            Farnum, Franklyn  &              Tamiroff, Akim  &              Tamiroff, Akim \\ 
\textsc{4th} &          Miller, Harold (I)  &          Miller, Harold (I)  &            Farnum, Franklyn  &               Welles, Orson \\ 
\textsc{5th} &                 Chefe, Jack  &              Holmes, Stuart  &          Miller, Harold (I)  &              Sayre, Jeffrey \\ 
\textsc{6th} &              Holmes, Stuart  &              Sayre, Jeffrey  &              Sayre, Jeffrey  &          Miller, Harold (I) \\ 
\textsc{7th} &               Steers, Larry  &                 Chefe, Jack  &          Quinn, Anthony (I)  &            Farnum, Franklyn \\ 
\textsc{8th} &               Par\`is, Manuel  &               Par\`is, Manuel  &         O'Brien, William H.  &             Kemp, Kenner G. \\ 
\textsc{9th} &         O'Brien, William H.  &         O'Brien, William H.  &              Holmes, Stuart  &          Quinn, Anthony (I) \\ 
\textsc{10th} &               Sayre, Jeffrey  &            Stevens, Bert (I)  &            Stevens, Bert (I)  &          O'Brien, William H. \\ 
\hline
\multicolumn{5}{c}{} \\ 
\hline
\textsc{Year} & \textbf{1980} & \textbf{1985} & \textbf{1990} & \textbf{1995}\\ 
\textsc{Nodes} & 310\,278 & 375\,322 & 463\,078 & 557\,373\\ 
\textsc{Edges} & 18\,252\,462 & 20\,970\,510 & 24\,573\,288 & 28\,542\,684\\ 
\textsc{Perf.~ratio} & 0.62\% & 0.45\% & 0.34\% & 0.26\%\\ 
\hline
\textsc{1st} &               Flowers, Bess  &               Welles, Orson  &               Welles, Orson  &        Lee, Christopher (I) \\ 
\textsc{2nd} &            Harris, Sam (II)  &               Flowers, Bess  &             Carradine, John  &               Welles, Orson \\ 
\textsc{3rd} &               Welles, Orson  &            Harris, Sam (II)  &               Flowers, Bess  &          Quinn, Anthony (I) \\ 
\textsc{4th} &              Sayre, Jeffrey  &          Quinn, Anthony (I)  &        Lee, Christopher (I)  &           Pleasence, Donald \\ 
\textsc{5th} &          Quinn, Anthony (I)  &              Sayre, Jeffrey  &            Harris, Sam (II)  &               Hitler, Adolf \\ 
\textsc{6th} &              Tamiroff, Akim  &             Carradine, John  &          Quinn, Anthony (I)  &             Carradine, John \\ 
\textsc{7th} &          Miller, Harold (I)  &             Kemp, Kenner G.  &           Pleasence, Donald  &               Flowers, Bess \\ 
\textsc{8th} &             Kemp, Kenner G.  &          Miller, Harold (I)  &              Sayre, Jeffrey  &             Mitchum, Robert \\ 
\textsc{9th} &            Farnum, Franklyn  &            Niven, David (I)  &               Tovey, Arthur  &            Harris, Sam (II) \\ 
\textsc{10th} &             Niven, David (I)  &               Tamiroff, Akim  &                Hitler, Adolf  &               Sayre, Jeffrey \\ 
\hline
\multicolumn{5}{c}{} \\ 
\hline
\textsc{Year} & \textbf{2000} & \textbf{2005} & \textbf{2010} & \textbf{2014}\\ 
\textsc{Nodes} & 681\,358 & 880\,032 & 1\,237\,879 & 1\,797\,446\\ 
\textsc{Edges} & 33\,564\,142 & 41\,079\,259 & 53\,625\,608 & 72\,880\,156\\ 
\textsc{Perf.~ratio} & 0.22\% & 0.18\% & 0.19\% & 0.14\%\\ 
\hline
\textsc{1st} &        Lee, Christopher (I)  &               Hitler, Adolf  &               Hitler, Adolf  &         Madsen, Michael (I) \\ 
\textsc{2nd} &               Hitler, Adolf  &        Lee, Christopher (I)  &        Lee, Christopher (I)  &                Trejo, Danny \\ 
\textsc{3rd} &           Pleasence, Donald  &                Steiger, Rod  &              Hopper, Dennis  &               Hitler, Adolf \\ 
\textsc{4th} &               Welles, Orson  &      Sutherland, Donald (I)  &          Keitel, Harvey (I)  &           Roberts, Eric (I) \\ 
\textsc{5th} &          Quinn, Anthony (I)  &           Pleasence, Donald  &            Carradine, David  &             De Niro, Robert \\ 
\textsc{6th} &                Steiger, Rod  &              Hopper, Dennis  &      Sutherland, Donald (I)  &               Dafoe, Willem \\ 
\textsc{7th} &             Carradine, John  &          Keitel, Harvey (I)  &               Dafoe, Willem  &          Jackson, Samuel L. \\ 
\textsc{8th} &      Sutherland, Donald (I)  &          von Sydow, Max (I)  &          Caine, Michael (I)  &          Keitel, Harvey (I) \\ 
\textsc{9th} &             Mitchum, Robert  &          Caine, Michael (I)  &               Sheen, Martin  &            Carradine, David \\ 
\textsc{10th} &                Connery, Sean  &                Sheen, Martin  &                    Kier, Udo  &         Lee, Christopher (I) \\
\hline
\end{tabular}
\end{table*}
\end{small}

\paragraph{The Algorithm}
The results outline that the performance ratio decreased drastically while the graph size increased: this suggests that the performances of our algorithm increase with the input size. 
This is even more visible from \Cref{fig:actorperformance}, where we have plotted the inverse of the performance ratio with respect to the number of nodes. It is clear that the plot is very close to a line, especially if we exclude the last two values. This means that the performance ratio is close to $\frac{1}{cn}$, where $c$ is the slope of the line in the plot, and the total running time is well approximated by $\frac{1}{cn}\O(mn)=\O(m)$. This means that, in practice, for the IMDB actor graph, our algorithm is linear in the input size.

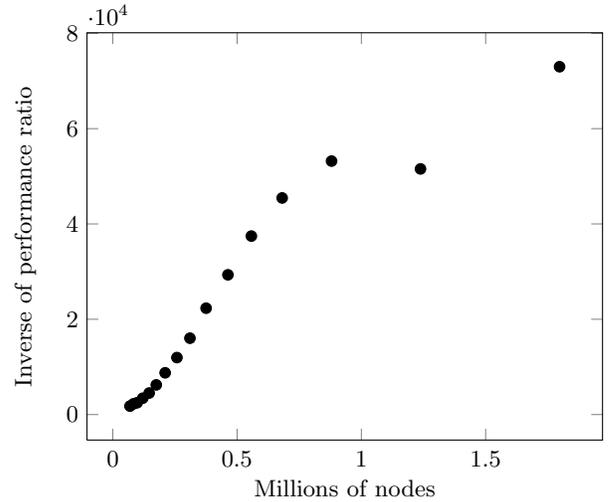
\begin{figure}
\centering
\begin{tikzpicture}
\tikzset{
did/.style={rectangle,draw=none,fill=none,inner sep=0cm},
every node/.style={did}}

\pgfplotsset{cycle list={{mark=*}},height=7cm,width=\columnwidth}
\begin{axis}[xlabel=Millions of nodes,
			ylabel=Inverse of performance ratio,
			y label style={did,at={(0.06,0.5)}},
			x tick label style={did, minimum height = 0.5cm},
			x label style={did},
			y tick label style={did, inner sep=3pt}]

\addplot+[mark=*,only marks] coordinates {
(.069011,1777.1459036787)
(.083068,2255.299954894)
(.097824,2482.005460412)
(.120430,3434.0659340659)
(.146253,4520.7956600362)
(.174826,6226.6500622665)
(.210527,8764.2418930763)
(.257896,11961.7224880383)
(.310278,16025.641025641)
(.375322,22321.4285714286)
(.463078,29325.5131964809)
(.557373,37453.1835205992)
(.681358,45454.5454545455)
(.880032,53191.4893617021)
(1.237879,51546.3917525773)
(1.797446,72992.700729927)
};
\end{axis}
\end{tikzpicture}
\caption{growth of performance ratio with respect to the number of nodes.}
\label{fig:actorperformance}
\end{figure}

\paragraph{The Results} In 2014, the most central actor is Michael Madsen, whose career spans 25 years and more than 170 films. Among his most famous appearances, he played as \emph{Jimmy Lennox} in \emph{Thelma \& Louise} (Ridley Scott, 1991), as \emph{Glen Greenwood} in \emph{Free Willy} (Simon Wincer, 1993), as \emph{Bob} in \emph{Sin City} (Frank Miller, Robert Rodriguez, Quentin Tarantino), and as \emph{Deadly Viper Budd} in \emph{Kill Bill} (Quentin Tarantino, 2003-2004). It is worth noting that he played in movies of very different kinds, and consequently he could ``reach'' many actors in a small amount of steps.
The second is Danny Trejo, whose most famous movies are  %Heat (Michael Mann, 1995), Desperado (Robert Rodriguez, 1995), and Con Air (Simon West, 1997), 
\emph{Heat} (Michael Mann, 1995), where he played as \emph{Trejo}, \emph{Machete} (Ethan Maniquis, Robert Rodriguez, 2010) and \emph{Machete Kills} (Robert Rodriguez, 2013), where he played as \emph{Mathete}.
The third ``actor'' is not really an actor: he is the German dictator Adolf Hitler: he was also the most central actor in 2005 and 2010, and he was in the top-10 since 1990. This a consequence of his appearances in several archive footages, that were re-used in several movies (he counts 775 credits, even if most of them are in documentaries or TV-shows, that were eliminated). Among the movies he is credited in, we find \emph{Zelig} (Woody Allen, 1983), and \emph{The Imitation Game} (Morten Tyldum, 2014): obviously, in both movies, he played himself.

Among the other most central actors, we find many people who played a lot of movies, and most of them are quite important actors. However, this ranking does not discriminate between important roles and marginal roles: for instance, the actress Bess Flowers is not widely known, because she rarely played significant roles, but she appeared in over 700 movies in her 41 years career, and for this reason she was the most central for 30 years, between 1950 and 1980. Finally, it is worth noting that we never find Kevin Bacon in the top 10, even if he became famous for the ``Six Degrees of Kevin Bacon'' game \url{http://oracleofbacon.org}, where the player receives an actor $x$, and he has to find a path of length at most $6$ from $x$ to Kevin Bacon in the actor graph. Kevin Bacon was chosen as the goal because he played in several movies, and he was thought to be one of the most central actors: this work shows that, actually, he is quite far from being in the top 10. Indeed, his closeness centrality is $0.336$, while the most central actor, Michael Madsen, has centrality $0.354$, and the 10th actor, Christopher Lee, has centrality $0.350$. We have run again the algorithm with $k=100$, and we have seen that the 100th actor is Rip Torn, with closeness centrality $0.341$.

\bibliographystyle{plain}
\bibliography{library}
\end{document}